\documentclass[11pt,letterpaper]{article}

\usepackage[margin=1in]{geometry}
\usepackage{hyperref}
\usepackage{algorithm}
\usepackage{algorithmic}

\usepackage{hyperref,enumitem}
\hypersetup{
	colorlinks=true,
	citecolor = blue       %
}
\setlist{noitemsep,leftmargin=\parindent,topsep=2pt}
\setlist{noitemsep,topsep=2pt}
\usepackage[numbers]{natbib}
\usepackage{url}            %
\usepackage{xspace}
\usepackage{graphicx}
\usepackage{amsmath,amssymb,amsthm}
\usepackage{booktabs} %
\usepackage{color}
\usepackage[font=small,labelfont=bf]{caption}
\usepackage{subcaption}
\usepackage{multirow}
\usepackage{authblk}

\usepackage{comment}

\usepackage{tikz}
\usetikzlibrary{arrows}
\usetikzlibrary{shapes}
\usetikzlibrary{calc}

\newcount\Comments 
\newcommand{\kibitz}[2]{\ifnum\Comments=1{\color{#1}{#2}}\fi}

\newcommand{\E}{\mathbb{E}}
\DeclareMathOperator*{\argmin}{arg\,min}

\newcommand{\1}{\mathbb{I}}

\newcommand{\Var}{\mathrm{Var}}
\newcommand{\Bias}{\mathrm{Bias}}
\newcommand{\eps}{\epsilon}

\newcommand{\R}{\mathbb{R}}
\newcommand{\B}{\mathcal{B}}
\newcommand{\F}{\mathcal{F}}

\newcommand{\PC}{\mathcal{P}_{\mathcal{C}}}
\newcommand{\PiC}{\Pi_{\mathcal{C}}}
\newcommand{\PP}{\mathbb{P}}

\renewcommand{\eps}{\varepsilon}

\theoremstyle{plain}
\newtheorem{theorem}{Theorem}[section]
\newtheorem{corollary}{Corollary}[section]
\newtheorem{lemma}[theorem]{Lemma}
\newtheorem{definition}[theorem]{Definition}

\begin{document}

\title{Reserve Price Optimization for First Price Auctions}

\author[a]{Zhe Feng\thanks{Supported by a NSF award CCF-1841550 and a Google PhD Fellowship. This work was done when the first author was an intern in Google Inc, NYC.}}
\author[b]{S\'ebastien Lahaie}
\author[b]{Jon Schneider}
\author[b]{Jinchao Ye}

\affil[a]{John A.~Paulson School of Engineering and Applied Sciences, Harvard University \authorcr \texttt{zhe\_feng@g.harvard.edu}}
\affil[b]{Google Inc, NYC \authorcr \texttt{slahaie, jschnei, jinchao@google.com}}

\date{June 28, 2020}

\maketitle

\begin{abstract}
The display advertising industry has recently transitioned from second- to first-price auctions as its primary mechanism for ad allocation and pricing. In light of this, publishers need to re-evaluate and optimize their auction parameters, notably reserve prices. In this paper, we propose a gradient-based algorithm to adaptively update and optimize reserve prices based on estimates of bidders' responsiveness to experimental shocks in reserves. Our key innovation is to draw on the inherent structure of the revenue objective in order to reduce the variance of gradient estimates and improve convergence rates in both theory and practice. We show that revenue in a first-price auction can be usefully decomposed into a \emph{demand} component and a \emph{bidding} component, and introduce techniques to reduce the variance of each component. We characterize the bias-variance trade-offs of these techniques and validate the performance of our proposed algorithm through experiments on synthetic data and real display ad auctions data from %
Google ad exchange.
\end{abstract}

\section{Introduction}\label{sec:intro}
A reserve price in an auction specifies a minimum acceptable winning bid, below which the item remains with the seller. The reserve price may correspond to some outside offer, or the value of the item to the seller itself, and more generally may be set to maximize expected revenue~\citep{Myerson1981}. In a data-rich environment like online advertising auctions it becomes possible to learn a revenue-optimal reserve price over time, and there is a substantial literature on optimizing reserve prices for \emph{second-price} auctions, which have been commonly used to allocate ad space~\citep{paes2016field,MM2016,MV2017}.

In this work we examine the problem of reserve price optimization in \emph{first-price} (i.e., pay-your-bid) auctions, motivated by the fact that all the major ad exchanges have recently transitioned to this auction format as their main ad allocation mechanism~\citep{digidayFirstprice,BiglerFP}. First-price auctions have grown in favor because they are considered more transparent, in the sense that there is no uncertainty in the final price upon winning~\citep{digidayTransparency}.\footnote{The full reasons for the transition are complex, and include the rise of ``header bidding''~\cite{digidayHeader}. A header bidding auction is a first-price auction usually triggered by code in a webpage header (hence the name).}
Unless restrictive assumptions are met, there is in theory no revenue ranking between first- and second-price auctions~\citep{krishna2009auction}, and there is no guarantee that reserve prices optimized for second-price auctions will continue to be effective in a first-price setting.%

From a learning standpoint the shift from second- to first-price auctions introduces several new challenges. In a second-price auction, truthful bidding is a dominant strategy no matter what the reserve. The bidders' value distributions are therefore readily available, and bids stay static (in principle) as the reserve is varied.
In a first-price auction, in contrast, bidders have an incentive to shade their values when placing their bids, and bid-shading strategies can vary by bidder. The gain from setting a reserve price now comes if (and only if) it induces higher bidding, so an understanding of bidder responsiveness becomes crucial to setting effective reserves.

Bid adjustments in response to a reserve price can occur at different timescales. If a bidder observes that it wins too few auctions because of the reserve price, it may increase its bid in the long-term (in a matter of hours up to weeks). Our focus here is on setting reserves prices by taking into account \emph{immediate} bidder responses to reserves. 
We assume that each bidder has a fixed, unknown bidding function $b(r, v)$ that depends on its private value $v$ and the observed auction reserve $r$. This agrees with practice in display ad auctions because the reserve $r$ is normally sent out in the 'bid request' message to potential bidders~\citep{OpenRTB}. To the extent that the bid function responds to $r$, first-price reserves can potentially show an immediate positive effect on revenue.
\vspace{-5pt}
\subsection*{Our Results}
\vspace{-5pt}
We propose a gradient-based approach to adaptively improve and optimize reserve prices, where we perturb current reserves upwards and downwards (e.g., by 10\%) on random slices of traffic to obtain gradient estimates. 

Our key innovation is to draw on the inherent structure of the revenue objective in order to reduce the variance of gradient estimates and improve convergence rates in both theory (e.g., see Corollary~\ref{cor:convergence-rate-bid-truncation}) and practice. We show that revenue in a first-price auction can be usefully decomposed into two terms: a \emph{demand curve} component which depends only on the bidder's value distribution; and a \emph{bidding} component whose variance can be reduced based on natural assumptions on bidding functions. 

A demand curve is a simpler, more structured object than the original revenue objective (e.g., it is downward-sloping), so the demand component lends itself to parametric modeling to reduce the variance.
We offer two variance reduction techniques for the bidding component\footnote{Variance reduction of the bidding component relies on the insight that bids far above the reserves are little affected by them (under natural bidding models), so these bids can be filtered out when computing gradient estimates---changes in such bids are likely due to noise rather than any effect of reserves.}, referred to as \emph{bid truncation} and \emph{quantile truncation}. Bid truncation can strictly decrease variance with no additional bias assuming the right bidding model, whereas quantile truncation may introduce bias but is less sensitive to assumptions on the bidding model.

We evaluate our approach over synthetic data where bidder values are drawn uniformly, and also over real bid distributions collected from the logs of the Google ad exchange with different bidder response models.
Our experimental results confirm that the combination of variance reduction on both objective components leads to the fastest convergence rate. %
For the demand component, a simple logistic model works well over the synthetic (i.e., uniform) data, but a flexible neural net is needed over the semi-synthetic data. For the bidding component, we find that quantile truncation is much more robust to assumptions on the bidding model.
\vspace{-5pt}
\subsection*{Related Work}
\vspace{-5pt}
This paper connects with the rich literature on \emph{reserve price optimization for auctions}, e.g., \citep{Myerson1981, riley1981optimal}. %
How to set optimal reserve prices in \emph{second price auctions} based on access to bidders' historical bid data has been an increasingly popular research direction in Machine Learning community, e.g.,~\citep{OS2011, MM2016, MV2017}. Another related line of work uses no-regret learning in second price auctions with partial information feedback to optimize reserve prices, e.g.,~\citep{Blum2004, CGM2015}. All of the works cited so far rely on the fact that the seller can directly learn the valuation distribution from historical bid data, since the second price auction is truthful. 

For first-price auctions, we have found little work on setting optimal reserves for asymmetric bidders, since there are no characterizations of equilibrium strategies for this case. Results are only available for limited environments, such as bidders with uniform valuation distributions~\citep{krishna2009auction, Matthews1995}. %
Recently, there has been a line of work regarding \emph{revenue optimization against strategic bidders in repeated auctions}, e.g.,~\citep{Amin13,Huang18}. In this paper, instead of assuming bidders act strategically, we assume each bidder has a fixed bidding function in response to reserves. This is a common assumption in large market settings and in the dynamic pricing literature \citep{Mao18}. 

The algorithms developed in this paper are related to the literature on \emph{online convex optimization with bandit feedback}~\citep{FKM2005, HL2014, ADX2010, AFH2011}.
However, there are two key differences with our work: (1) the revenue function in a first price auction is non-convex, and (2) the seller cannot obtain perfect revenue feedback under perturbed reserves with just a single query (i.e., auction)---the seller needs multiple queries to achieve accurate estimates with high confidence. Our algorithm is also related to \emph{zeroth-order stochastic gradient methods} \citep{GL2013, BG2018, Ghadimi2019, LLCHA2018}, which we discuss in detail later in Section~\ref{sec:method}.%
\section{Preliminaries}\label{sec:prelim}

We consider a setting where a seller repeatedly sells a single item to a set of $m$ bidders via a first price auction. In such an auction, the seller first sends out a reserve price $r$ to all bidders. Each bidder $i$ then submits a bid $b_i$.  The bidder with the highest bid larger than $r$ wins the item and pays their bid; if no bidder bids above $r$, the item goes unallocated. Note that the type of reserve price we consider in this work is \emph{anonymous} in the sense that each bidder sees the same reserve price.

Each bidder $i$ has a private valuation $v_i \in [0, 1]$ for the item, where each value $v_i$ is drawn independently (but not necessarily identically) from some unknown distribution.\footnote{This is without loss of generality, our analysis can easily be applied to any bounded valuation setting.} In a first-price auction, only the highest bid matters for both allocation and pricing. Thus, to simplify the notation, we write $v = \max_i v_i$ to denote the maximum value and $v$ is drawn i.i.d.\ from an unknown distribution $F$ across each auction. %
Our analysis from here on will refer to this `representative' highest bidder. (See Appendix \ref{sec:multi-to-one} for a rigorous justification of why we can reduce multiple bidders to a single bidder.)

We write $b(r, v)$ to denote the maximum bid when the reserve price is $r$ and the maximum value is $v$, and $\B(r)$ to denote the distribution of $b(r, v)$ for a fixed $r$ when $v$ is drawn according to $\F$. 
The main goal of the seller considered in this work is to learn the optimal reserve price $r \in [0,1]$ that maximizes expected revenue:
\begin{equation}\label{eq:main-objective}
\E_{v \sim \F} \left[b(r, v)\cdot \1\{b(r, v) \geq r\}\right].
\end{equation}
Note that there is no reason for a bidder to bid a positive value less than the reserve $r$: such a bid is guaranteed to lose. Therefore, without loss of generality we can assume that if $b(r, v) < r$, then $b(r, v) = 0$. This allows us to write the revenue simply as:
$$\mu(r) = \E_{b \sim \B(r)} \left[ b \right] = \E_{v \sim \F} \left[b(r, v)\right].$$
In this paper, we focus on maximizing the revenue function $\mu(r)$ in the steady state. %

\vspace{-5pt}
\subsection*{Response Models}\label{sec:response-model}
\vspace{-5pt}

We begin by describing some general properties of bidding functions that hold for any utility-maximizing bidders (see \citep{Matthews1995} for further discussion). 

\begin{definition}
\label{def:fp-property}
A \emph{bidding function} $b(r, v)$ satisfies the following properties: 1) $b(r, v) \leq v$ for all $v$; 2) $b(r, v) \geq r$ for $v \geq r$; 3) $b(r, v) = 0$ for $v < r$; 4) $b(r, v)$ is non-decreasing in $v$ for all $r$. %
\end{definition}

In some of our algorithms, we would like to impose additional constraints on the response model which, while not a consequence of utility-maximizing behavior, are likely to hold in practice. One such constraint is the \textit{diminishing sensitivity in value of bid to reserve}. This says that bidders with a larger value will change their bid less in response to a change in reserves.

\begin{definition}[Diminishing sensitivity of bid to reserve]
If $v_H > v_L \geq r$, then for $\delta > 0$ and $v_L \geq r+\delta$ we have $b(r + \delta, v_H) - b(r, v_H) \leq b(r + \delta, v_L) - b(r, v_L)$.

\end{definition}

One natural and concrete example of a response model is a bidder that increases its bid to the reserve as long as the reserve is below its value. We refer to this as the \emph{perfect response model}, formally defined as follows.

\begin{definition}
\label{def:perfect-response-model}
A \emph{perfect response} bidding function takes the form:
$$
b(r, v) = \left\{
\begin{array}{cl}
b(0, v) & \mbox{if $b(0, v) \geq r$}  \\
r & \mbox{if $b(0, v) < r \leq v$} \\
0  & \mbox{if $v < r$} \\
\end{array}
\right.
$$
\end{definition}
\noindent
Note that the perfect response model is based on the \emph{original bid} of the bidder under reserve price $0$, namely $b(0, v)$. If $b(0, v)$ is already above the reserve, then this bidder is unaffected by the reserve. %
Note that the perfect response model satisfies the diminishing sensitivity property.

In practice, bidders are unlikely to exactly follow the perfect response model; for example, bidders will often increase their bid to some amount strictly above the reserve $r$ so as to remain competitive with other bidders. For this reason, we propose a relaxation of the perfect response model which we call the \emph{$\eps$-bounded response model}: the bid is at most $\eps$ greater than what it would have been under the perfect response model if $b(0, v) < r \leq  v$ (see also Definition~\ref{def:eps-bounded-response-model}). Note that the $\eps$-bounded response model becomes the perfect response model when $\eps = 0$.

\vspace{-5pt}
\section{Gradient Descent Framework}
\label{sec:method}
\vspace{-5pt}

The first-price auction setting introduces several challenges for setting reserve prices. First, the seller cannot observe true bidder values because truthful bidding is not a dominant strategy in a first-price auction. Second, how the bidders will react to different reserves is unknown to the seller---the only information that the seller receives is bids drawn from distribution $\B(r)$ when the seller sets a reserve price $r$.

One natural idea, and the approach we take in this paper, is to optimize the reserve price via gradient descent. Gradient descent is only guaranteed to converge to the optimal reserve when our objective is convex (or at least, unimodal), which is not necessarily true for an arbitrary revenue function. However, gradient descent has a number of practical advantages for reserve price optimization, including:
\begin{enumerate}
\item  Gradient descent allows us to incorporate prior information we may have about the location of a good reserve price (possibly significantly reducing the overall search cost).
\item The adaptivity of gradient descent allows us to quickly converge to a local optimum and follow this optimum if it changes over time, significantly saving on search cost (over global methods such as grid search). 
\item In practice, many revenue curves have a unique local optimum (see Section \ref{sec:experiment}), so gradient descent is likely to converge to the optimal reserve.
\end{enumerate}

More specifically, since the seller has no direct access to the gradients (i.e, first-order information) of $\mu(r)$, we consider approaches that fit in the framework of \textit{zeroth-order stochastic optimization}. %
Our framework, summarized in~Algorithm \ref{alg:OSGD-revenue-general}, proceeds in rounds. In round $t$ where the current reserve is $r_t$, the seller selects a perturbation size $\beta_t$ and randomly sets the reserve price to either $(1 + \beta_t)r_t$ or $(1-\beta_t)r_t$ on separate slices of experiment traffic, until it has received $n_t$ samples from both $\B((1+\beta_t)r_t)$ and $\B((1-\beta_t)r_t)$. The seller then uses these $2n_t$ samples to estimate the gradient $\hat{G}_t$ of the revenue curve $\mu(r)$ at $r_t$ and updates the reserve price based on this gradient estimate using learning rate (step size) $\alpha_t$.  %

We assume that we have access to a fixed total number of samples $N=\sum_{t=1}^T n_t$ (the number of iterations $T$ is a variable that will be fixed later). There is then a trade-off between $n_t$ (i.e, the number of samples per iteration) and $T$ (the number of iterations available to optimize the reserve price).

\begin{algorithm}[tb]
\caption{Zeroth-order stochastic projected gradient framework for reserve optimization.} \label{alg:OSGD-revenue-general}
\begin{algorithmic}
\STATE {\bf Input:} Initial reserve $r_1 \in (0,1)$, total number of iterations $T$ (a variable to be fixed later).
\STATE {\bf Output:} Reserve prices $r_2, r_3,\dots, r_{T+1}$.
\FOR{$t=1,2,\dots,T$}
\STATE Set a reserve price of $r^{+}_t = (1+\beta_t)r_{t}$ in $n_t$ %
auctions.
\STATE
Set a reserve price of $r^{-}_t = (1-\beta_t)r_{t}$ in $n_t$ %
auctions.
\STATE
Construct an estimate $\hat{G}_t$ of the gradient of revenue at $r_{t}$, based on the feedback of experiments.
\STATE Update reserve: $r_{t+1} = \Pi(r_t + \alpha_t \hat{G}_t)$, where
$$\Pi(x) = \argmin_{z\in(0, 1)}|z-x|.$$
\vspace{-10pt}
\ENDFOR
\end{algorithmic}
\end{algorithm}

Zeroth-order stochastic gradient descent is a well-studied problem~\citep{GL2013, BG2018, Ghadimi2019, LLCHA2018}. 
In this paper, we focus on taking advantage of the structure of $b(r, v)$ to construct good \emph{discrete gradient} estimates $\hat{G}_t$, as this aspect is specific to the problem of reserve price optimization. Specifically, we tackle the following problem which we term the \textit{discrete gradient problem}:

\begin{itemize}
\item \textbf{Input: } $n$ samples $X_1^+,\cdots, X_n^+$ drawn i.i.d from $\B(r^{+})$ and $n$ samples $X_1^,\cdots, X_n^-$ drawn i.i.d from $\B(r^{-})$, for known $r^{+} > r^{-}$.
\item \textbf{Output: } An estimator $\hat{G}$ %
for the discrete derivative $(\mu(r^{+}) - \mu(r^{-}))/(r^{+} - r^{-})$. This estimator has bias $\Bias(\hat{G})$ and variance $\Var(\hat{G})$, where
$\Bias(\hat{G}) = \left|\E[\hat{G}] - \frac{\mu(r^{+}) - \mu(r^{-})}{r^{+} - r^{-}}\right|$.
\end{itemize}

Solutions to the discrete gradient problem with small bias and variance directly translate into faster convergence rates for our gradient descent. %
We provide a detailed convergence result %
in Theorem~\ref{thm:convergence-rate} in Appendix~\ref{app:convergence-rate-proof}. We summarize this result informally as follows.

\begin{theorem}[Informal Restatement of Theorem \ref{thm:convergence-rate}]\label{thm:informal-convergence-rate}
If for all $t$, $\Bias(\hat{G}_t) \leq B$ and $\Var(\hat{G}_t) \leq V$ then for optimal choices of $\alpha_t$ and $n_t$ (and fixing $\beta_t = \delta/2r_t$), Algorithm \ref{alg:OSGD-revenue-general} satisfies
$$\min_{t\in[T]} |\PC^t|^2 = \tilde{\mathcal{O}}\left(T^{-1/2} + \delta^2 + B^2 + V + (T/N)^2\right).$$
Here $\PC^t$ can be thought of as the true gradient at round $t$ (see Definition \ref{def:gradient_mapping} in Appendix). 
\end{theorem}

Intuitively, we want to design an estimator and choose our parameters $\alpha_t, \beta_t, n_t$, so as to trade off between $\delta$, $B$, and $V$. In the following sections, we show how to do this for a variety of bidder response models.

\vspace{-5pt}
\subsection*{Naive Gradient Estimation}
\label{sec:naive-gradient-descent}
\vspace{-5pt}

The simplest method for estimating the discrete gradient is to take the difference between the average revenue from bids from $\B(r^{+})$ and the average revenue from bids from $\B(r^{-})$.
More formally, %
we compute discrete gradient as,
{\small
\begin{eqnarray}\label{eq:naive-gd}
\hat{G} = \frac{\sum_{i = 1}^{n} X^{+}_i - \sum_{i = 1}^{n} X^{-}_i}{n(r^{+} - r^{-})}.
\end{eqnarray}
We show that $\hat{G}$ has the following properties.
}
\begin{theorem}\label{thm:discrete-gradient-revenue}
Assume that $r^{+} - r^{-} = \delta$, then %
$\Bias(\hat{G}) = 0, \Var(\hat{G}) \leq \frac{1}{2\delta^2 n}$.
\end{theorem}

This leads to the following convergence rate via Theorem \ref{thm:informal-convergence-rate}. %

\begin{corollary}\label{cor:convergence-rate-naive-estimator}
Using this estimator $\hat{G}$, and setting $T = N^{1/2}$ and $\delta = \Theta(N^{-1/8})$, Algorithm \ref{alg:OSGD-revenue-general} achieves convergence,
$\min_{t\in[T]} |\PC^t|^2 \leq \widetilde{\mathcal{O}}\left(N^{-1/4}\right).$
\end{corollary}
Although there are no matching lower bounds, this is the best known asymptotic convergence rate for zeroth-order optimization over a non-convex objective~\citep{GL2013, BG2018}.
The naive gradient estimation approach has the advantage that it works regardless of response model, is simple to compute (it uses only revenue information and not individual bids), and leads to an \textit{unbiased} estimator for the discrete derivative. The disadvantage is that the variance of this estimator can be large (especially as we take $\delta$ small). In the following section, we show how to address this by taking into account the inherent structure of the revenue objective based on an underlying bidder response model.

\vspace{-5pt}
\section{Variance Reduced Gradient Estimation}
\label{sec:gradient-descent-decomposed-revenue}
\label{sec:revenue-decomposition}
\vspace{-5pt}

In this section, we first introduce another representation of the revenue formula by decomposing it into a \emph{demand} component and a \emph{bidding} component. We then propose techniques to reduce the variance of the discrete gradient of each component.

\subsection{Revenue Decomposition}
We can decompose the revenue $\mu(r)$ in the following way.

\begin{theorem}\label{thm:revenue-decomposition}
We have that
\begin{equation}\label{eq:revenue-decomposition}
\mu(r) = \E_{v\sim \F}[\max(b(r, v) - r, 0)] + r \Pr_{v \sim \F}[v \geq r].
\end{equation}
\end{theorem}
\vspace{-10pt}

Define $E(r) = \E_{v \sim \F}[\max(b(r, v) - r, 0)]$ and $D(r) = \Pr_{v \sim \F}[v \geq r]$, so that $\mu(r) = E(r) + rD(r)$. These two terms capture two different aspects of bidder behavior which contribute to revenue. The function $D(r)$ amounts to a ``demand curve'' which gives the proportion of values that clear the reserve $r$, and therefore the proportion of auctions that are bid on at $r$. If the auction were just a simple posted-price auction (i.e., the winner is charged the quoted price $r$), then the \emph{demand} component $r D(r)$ would be the associated revenue.
However, in a first-price auction the winning bidder pays its bid, not the reserve. Therefore the \emph{bidding} component $E(r)$ captures the excess contribution from bids greater than the reserve.

To construct a good estimator $\hat{G}$ for the discrete gradient of $\mu(r)$, it suffices to construct good estimators $\hat{G}_{E}$ and $\hat{G}_{D}$ for the discrete gradients of $E(r)$ and $rD(r)$ respectively, and then output $\hat{G} = \hat{G}_{E} + \hat{G}_{D}$. Note that $\Bias(\hat{G}) \leq \Bias(\hat{G}_{E}) + \Bias(\hat{G}_{D})$ and $\Var(\hat{G}) \leq 2(\Var(\hat{G}_{D}) + \Var(\hat{G}_{E}))$, so it suffices to bound the bias and variance of each component separately.

\subsection{Estimating the Demand Component Gradient}

We begin by discussing how to estimate the gradient $\hat{G}_{D}$ of the demand component of revenue. %
One method of doing so is by estimating $D(r)$ with a parametric function $f_{\theta}(r)$, and using this approximation to estimate the gradient $\hat{G}_D$. (See Appendix~\ref{app:discussion-demand-function} for additional justification for why this is likely to be possible and helpful). %
Suppose that we have access to additional historical data $\mathcal{S}$ %
with which we can fit our parametric class to $D(r)$;  let $\hat{\theta}$ be the resulting learned parameter. This learned demand function gives rise to the following estimator $\hat{G}_D$:
\begin{equation}\label{eq:estimator-demand-gradient}
\hat{G}_D = \frac{r^+ f_{\hat{\theta}}(r^+) - r^- f_{\hat{\theta}}(r^-)}{r^+ - r^-}
\end{equation}

Note that this decreases overall variance, the variance of $\hat{G}_D$ is 0 because the randomness of $\hat{G}_D$ only comes from historical samples $\mathcal{S}$, which are independent of the samples obtained in the current round, at the cost of a possible increase in bias (due to inaccuracy in estimating $D(r)$). 

\subsection{Estimating the Bidding Component Gradient}
\label{sec:VR}
In this section we propose a variance reduction method to achieve a better estimator for $\hat{G}_E$ for a variety of bidder models.
\vspace{-10pt}
\paragraph{Variance reduction via bid truncation.}

We first consider the special case of the perfect response (and more generally, the $\eps$-bounded response) bidding model. In the perfect response model, if you were going to bid $b > r^{+}$ when the reserve was $r^{+}$, you will bid the same bid $b$ when the reserve is $r^{-}$. This means that large bids (bids larger than $r^{+}$) do not contribute in expectation to $\mu(r^{+}) - \mu(r^{-})$, but they do add noise to our gradient estimation. By filtering these out, we can reduce the variance of our estimator while keeping our estimator unbiased.

Since we only apply this filtering when estimating the bidding component $E(r)$ but not the demand component $r D(r)$, we must be careful when implementing this. Note that a large bid $b > r^{+}$ contributes $b - r^{+}$ to $E(r^{+})$ and $b - r^{-}$ to $E(r^{-})$, and therefore $r^{+} - r^{-}$ to $E(r^{+}) - E(r^{-})$. We can therefore construct an unbiased estimator for $E(r^{+}) - E(r^{-})$ by computing the contribution of unfiltered bids ($b < r^{+}$) from both $\B(r^{+})$ or $\B(r^{-})$ and then adding $r^{+} - r^{-}$ for each filtered bid in $\B(r^{-})$ (or equivalently, each filtered bid in $\B(r^{+})$; under perfect response, the fraction of filtered bids is equal in both models in expectation). Note that every bid from $\B(r^{+})$ is either filtered or has excess $0$, so we can write this gradient $\hat{G}_E$ entirely in terms of bids from $\B(r^{-})$. Formally, we define truncated bid $Y_i^-$ as
{\small
$$
Y^{-}_i = \left\{
\begin{array}{cl}
\max(X^{-}_i - r^{-}, 0) & \mbox{if } X^{-}_i \leq r^{+}  \\
(r^+ - r^-) & \mbox{otherwise} 
\end{array}
\right.
$$
}
Our estimate for the gradient of $E(r)$ is then given by 
\begin{eqnarray}\label{eq:discrete-gradient-bid-truncation}
\hat{G}_E = -\frac{\sum_{i = 1}^{n} Y^{-}_i}{n(r^{+} - r^{-})}.
\end{eqnarray}

Since any bid in an $\eps$-bounded model only differs from one in the perfect response model by at most $\eps$, we can apply this same estimator to an $\eps$-bounded response model. The following theorem characterizes the bias and variance of the estimator for the $\eps$-bounded response model.%

\begin{theorem}\label{thm:bias-variance-eps-bounded}
Assume that $r^{+} - r^{-} = \delta$, then the estimator $\hat{G}_E$ in Eq.~(\ref{eq:discrete-gradient-bid-truncation}) for $\eps$-bounded response model, satisfies: $\Bias(\hat{G}_E) = \frac{2\eps}{\delta}, \Var(\hat{G}_E) \leq  \frac{1}{4n}$.
\end{theorem}
Note that the bias of estimator $\hat{G}_E$ is 0 for the perfect response model. %
 The complete proof is given in Appendix~\ref{app:eps-bounded-response}. Combining the above results for $\hat{G}_E$ and $\hat{G}_D$, we have the following improved convergence result for the $\eps$-bounded response model.
\begin{corollary}\label{cor:convergence-rate-bid-truncation}
Suppose $\Bias(\hat{G}_D) \leq \eps_D/\delta$. Using the estimator $\hat{G}_E$ proposed in Eq.~(\ref{eq:discrete-gradient-bid-truncation}) for the $\eps$-bounded response model, setting $T = N^{2/3}$ and $\delta = \Theta(\sqrt{\eps+\eps_D})$, Algorithm \ref{alg:OSGD-revenue-general} achieves convergence,
$\min_{t\in[T]} |\PC^t|^2 \leq \widetilde{\mathcal{O}}\left(\eps+\eps_D+N^{-1/3}\right)$.
\end{corollary}
For perfect response bidding models, the above convergence rate is strictly faster than the convergence rate of naive estimator in Corollary~\ref{cor:convergence-rate-naive-estimator} (state-of-the-art convergence rate for zeroth-order stochastic gradient descent), but with additional bias coming from demand estimation. However, we show this bias has practically negligible effect on the revenue in our experiments. 

\paragraph{Variance reduction via quantile truncation.}

In Eq.~(\ref{eq:discrete-gradient-bid-truncation}), we reduced the variance of $\hat{G}_E$ by truncating all bids at the fixed threshold of $t = r^{+}$. In general, this does not quite work: for bidder response models that are far from perfect response, this truncation can introduce a very large bias. Here we demonstrate one technique for constructing good estimators $\hat{G}_E$ as long as the bidding function $b(r,v)$ possesses diminishing sensitivity in value to reserve. %

Instead of truncating in \textit{bid space}, we will instead want to truncate in \textit{value space} to reduce the variance. Even though we cannot directly truncate by values, since $b(r, v)$ is monotonically increasing in $v$, quantiles of bids (e.g., of $\B(r^{+})$ and $\B(r^{-})$) directly correspond to quantiles of values (of $\F$). Instead of setting a threshold $t$ directly on the value, it is therefore equivalent to truncate at a \textit{fixed quantile of the bid distribution}. 

To achieve this, we first sort $X_i^+$ and $X_i^-$ in ascending order. Then we compute $\hat{G}_E$ as
\begin{eqnarray}\label{eq:discrete-gradient-quantile-truncation}
\hat{G}_E = \frac{\sum_{i = 1}^{qn} \max(X^{+}_i - r^{+}, 0) - \sum_{i = 1}^{qn} \max(X^{-}_i - r^{-}, 0)}{n(r^{+} - r^{-})} - (1-q),
\end{eqnarray}
where $q\in [0, 1]$ is the \emph{quantile threshold} used to truncate bids. The following theorem characterizes the bias and variance of the above $\hat{G}_E$,

\begin{theorem}\label{thm:quantile-vr}
Let $r^{+} - r^{-} = \delta$, $t = \F^{-1}(q)$, and $\tilde{t} = \F^{-1}(q + n^{-2/3})$. Then the estimator $\hat{G}_E$ in Eq.~(\ref{eq:discrete-gradient-quantile-truncation}) satisfies,
$\Bias(\hat{G}_E)\leq \frac{(1-q)(b(r^{+}, t) - b(r^{-}, t))}{\delta} + O(n^{-2/3}),
\Var(\hat{G}_E)\leq \frac{2\tilde{t}^2}{n\delta^2} + O(n^{-5/3}\delta^{-2})$.
\end{theorem}

Unlike with bid truncation, with quantile truncation we have a clear bias-variance tradeoff as we change $q$: larger values of $q$ decrease the bias (both by decreasing $(1-q)$ and $b(r^{+}, t) - b(r^{-}, t)$, which is decreasing due to diminishing sensitivity) but lead to larger variance. Since one can estimate this bound on the bias (by approximating $b(r^{+}, t) - b(r^{-}, t)$ via $Y^{+}_{qn} - Y^{-}_{qn}$), it is possible to choose $q$ to optimize this bias-variance tradeoff 
as one sees fit (for example, to minimize $B^2 + V$ in Theorem \ref{thm:informal-convergence-rate}). %
We show a convergence rate result for this \emph{quantile truncation} approach in Corollary~\ref{cor:convergence-rate-quantile-truncation} in Appendix~\ref{app:convergence-rate-quantile-truncation}.

\vspace{-10pt}
\section{Experiments}\label{sec:experiment}
\vspace{-5pt}
We evaluate the performance of our algorithms on synthetic and semi-synthetic data sets. %
Due to space limitations, we present the complete experimental results in Appendix~\ref{app:additional-experiment}.

\vspace{-5pt}
\subsection{Data Generation}
The data generation process consists of two parts: a base bid distribution specifying the distribution of bids when no reserve is set, and a response model describing how a bidder with bid $b$ would update its bid in response to a reserve of $r$.

{\bf Response models. } %
We assume that in the absence of a reserve bidders bid a constant fraction $\gamma$ of their value $v$ (i.e., $b = \gamma v$), which we refer to as \emph{linear shading}. We consider linear shading combined with perfect response and with $\epsilon$-bounded response, which we implement by adding a uniform $[0, \epsilon]$ random variable to the bid. We also examine equilibrium bidding for $n$ i.i.d.\ bidders with uniformly distributed valuation~\citep{krishna2009auction}: $b = \frac{r^n + (n-1)v^n}{n v^{n-1}}$.

{\bf Synthetic data. }
In our synthetic data sets, the (base) bid distribution is the uniform $[0,1]$ distribution. We apply the perfect response model, $\eps$-bounded response model and equilibrium bidding model. In the simulations, we apply a constant shading factor of $0.4$ for the perfect response model and $\eps$-bounded response model. For equilibrium bidding, we assume that each auction contains $n = 2$ bidders.%

{\bf Semi-synthetic data. }
For our semi-synthetic data sets, we separately collected the empirical distributions of winning bids over one day for 20 large publishers on a major display ad exchange. Each distribution was filtered for outliers and normalized to the interval $[0, 1]$. For this semi-synthetic data we only test the perfect-response model and $\eps$-bounded response model, since there is no closed-form solution for the equilibrium bidding strategy. We use 0.3 as the constant shading factor for semi-synthetic data.

\subsection{Methodology}

{\bf Gradient descent algorithms.} We examine five different gradient descent algorithms: (I) \emph{Naive GD}: naive gradient descent using the gradient estimator in Eq.~(\ref{eq:naive-gd}); (II) \emph{Naive GD with bid truncation}: gradient descent using the gradient estimator in Eq.~(\ref{eq:discrete-gradient-bid-truncation}) for the bidding component, and a naive estimate\footnote{We can form a naive unbiased estimator $\hat{G}_D = \frac{\hat{D}(r^+) - \hat{D}(r^-)}{r^+ - r^-}$, where $\hat{D}(r^+) = \frac{1}{n}\sum_i \1\{x^+_i \geq r^+\}$ and similarly for $\hat{D}(r^-)$.} of the demand component; 
(III) \emph{Naive GD with quantile truncation}: gradient descent using the gradient estimator in Eq.~(\ref{eq:discrete-gradient-quantile-truncation}) for the bidding component, and naive estimate of the demand component; 
(IV) \emph{Demand modeling with bid truncation}: Same as the second variant, but with a parametric model of the demand curve to estimate demand component of gradient;
(V) \emph{Demand modeling with quantile truncation}: Same as the third variant, but with a parametric model of the demand curve to estimate demand component of gradient.
The parameters used in these algorithms are specified in Appendix~\ref{app:additional-experiment}.

\begin{figure*}[t]
\begin{subfigure}[b]{0.33\textwidth}
\centering
\includegraphics[scale=0.33]{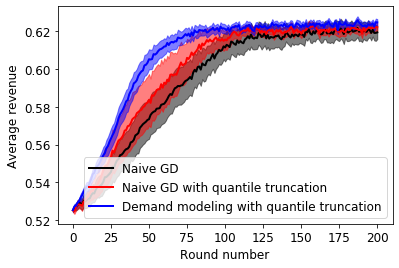}
\caption{Synthetic data with perfect response.}
\label{fig:synthetic-perfect-response-revenue}
\end{subfigure}
\begin{subfigure}[b]{0.33\textwidth}
\centering
\includegraphics[scale=0.33]{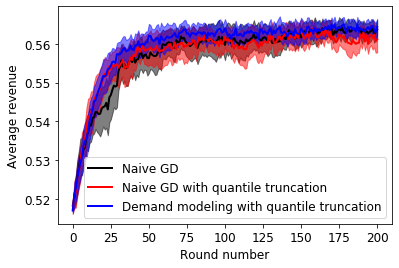}
\caption{Synthetic data with equilibrium response.}
\label{fig:synthetic-equilibrium-response-revenue}
\end{subfigure}
\begin{subfigure}[b]{0.33\textwidth}
\centering
\includegraphics[scale=0.33]{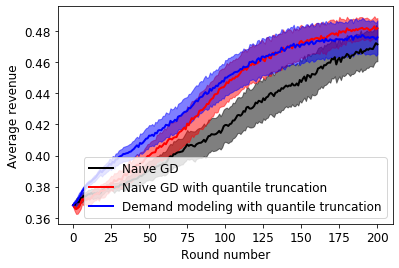}
\caption{Semi-synthetic data with perfect response.}
\label{fig:dremel-perfect-response-revenue}
\end{subfigure}
\caption{Revenue as a function of round $t$ for (a) synthetic data with perfect response, (b) synthetic data with equilibrium response, and (c) semi-synthetic data with perfect response.}
\label{fig:revenue-quantile}
\vspace{-15pt}
\end{figure*}

{\bf Demand curve estimation. } To reduce variance following the ideas of Section~\ref{sec:revenue-decomposition}, we need a model $\hat{G}_D$ for the demand component of the discrete gradient. Instead of estimating $\hat{G}_D$ from historical data, we adaptively learn the demand curve during the training process.
Concretely, at each round $t$, we observe new (reserve, demand) pairs from $2n_t$ samples and retrain our demand curve using all the samples observed up to the current round. We use this trained demand curve to compute $\hat{G}_D$ based on~\eqref{eq:estimator-demand-gradient}. For the synthetic data, a simple logistic regression can effectively learn the demand curve. However, the semi-synthetic data required a more flexible model so for this case we model demand using a fully connected neural network with 1 hidden layer, 15 hidden nodes and ReLU activations.

\begin{figure}[t]
\centering
\includegraphics[scale=0.4]{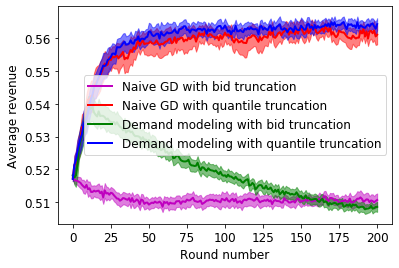}
\caption{Revenue as a function of round $t$ for synthetic data with equilibrium response.}%
\label{fig:failure-truncated-vr}
\vspace{-15pt}
\end{figure}

\subsection{Evaluation}

{\bf Effectiveness of gradient descent.}
First, we confirm that gradient descent can effectively find optimal reserves in our models. 
For each semi-synthetic model, we construct the revenue curve as a function of reserve %
with assumed response models. %
We find that 19 out of the 20 revenue curves have a clear single local maximum (the remaining curve has 2). In all cases (synthetic and semi-synthetic models), the revenue learned by the naive gradient descent algorithm is at least 95\% of the revenue at the optimal reserve, which indicates that gradient descent can efficiently find the optimal reserve in these cases despite the lack of convexity.

{\bf Effectiveness of variance reduction methods.}
We first evaluate the performance of the quantile-based variance reduction method. We run the algorithm variants (I), (III) and (V) under synthetic data and semi-synthetic data with multiple bidder response models. Figures~(\ref{fig:synthetic-perfect-response-revenue}) and~(\ref{fig:dremel-perfect-response-revenue}) show the revenue achieved by the three algorithms over time under the perfect response model. We find that quantile-based variance reduction leads to a more stable training process which converges faster than naive gradient descent. Figure~(\ref{fig:synthetic-equilibrium-response-revenue}) evaluates the performance of the three algorithm variants under synthetic data and an equilibrium response model, with similar conclusions. Overall, quantile-based variance reduction outperforms naive gradient descent. Moreover, with the addition of demand curve estimation, algorithm variant (V) achieves better revenue and converges to an optimal reserve faster than the other two algorithms, in agreement with our theoretical guarantees.

 We next consider variance reduction using bid truncation, which is used in algorithm variants (II) and (IV). Bid truncation is tailored to perfect response and performs the best overall for this response model, in accordance with the theoretical guarantees, but quantile truncation is competitive and often performs as well over the semi-synthetic data (see Appendix~\ref{app:additional-experiment} for a detailed comparison).
Under the equilibrium response model, bid truncation can in fact hinder the training process and lead to a substantially suboptimal reserve price (see Figure~\ref{fig:failure-truncated-vr}). In summary, quantile-based variance reduction coupled with a good demand-curve estimation is the method of choice to achieve good reserve prices under a range of different bid distributions and bidder response models.

\bibliography{rpo_first_price}
\bibliographystyle{plainnat}

\onecolumn
\newpage
\appendix
\appendix
\begin{center}
{
\Large
\textbf{
Reserve Price Optimization for First Price Auctions}
~\\
~\\	
Appendix
}
\end{center}

\section{From multiple bidders to a single bidder}\label{sec:multi-to-one}

Our auction can contain multiple bidders, each with their own value distribution $F_i$ and bid function $b_i(r, v_i)$. But when setting reserve prices, we only care about the maximum bid; more specifically, the distribution of maximum bid at each reserve. Thus it is useful to abstract away the set of bidders in the auction as a single ``mega-bidder'' whose value is the maximum of all the bidders' values and who always bids the maximum of all the bids. 

\begin{theorem}
Let $\F$ be the distribution of $\max(v_1, v_2, \dots, v_m)$ (where each $v_i$ is independently drawn from $F_i$) and let $\B(r)$ be the distribution of $\max(b_1(r, v_1), b_2(r, v_2), \dots, b_m(r, v_m))$. Then there exists a bid function $b(r, v)$ such that the distribution of $b(r, v)$ when $v \sim \F$ is equal to the distribution $\B(r)$. 
\end{theorem}
\begin{proof}
If $B_r(b)$ is the CDF of $\B(r)$ and $F(v)$ is the cdf of $\F$, let $b(r, v) = B_r^{-1}(F(v))$. This guarantees that if $v \sim \F$, then $b(r, v) \sim \B(r)$. 
\end{proof}

Note that this reduction also preserves the properties of Definition \ref{def:fp-property}. For example, if $b_i(r, v_i) \leq v_i$ for every bidder $i$, then the induced $b(r, v)$ also satisfies $b(r, v) \leq v$.

\section{Demand Function Estimation}\label{app:discussion-demand-function}
As in Section~\ref{sec:method}, it is possible to form a naive unbiased estimate of the demand component via the estimator $\hat{D}(r) = \frac{1}{n}\sum_{i=1}^{n}\1(X_i \geq r)$. The variance of the resulting unbiased estimator $\hat{G}_{D}$ is then bounded by (see Theorem \ref{thm:var_demand}), $\Var(\hat{G_D}) \leq  \frac{(r^{+})^2}{2\delta^2 n}$.

Note that for small $r$, the variance guarantee here is significantly better than the variance guarantee in Theorem~\ref{thm:discrete-gradient-revenue}. Thus, in instances where the optimal reserve is small (and hence we mostly test small $r^{+}$), combining this naive estimator with better estimators for $\hat{G}_E$ (like the ones we explore in the next section) can already lead to better convergence rates overall.

To obtain even better estimators, we can leverage the following two facts about the demand function. First, the demand function only depends on the value distribution $\F$ of the bidders, and not their specific bidding behavior. Since we expect value to be relatively stable in comparison to bidding behavior, this means that we can reasonably use data from previous rounds to learn the demand function and inform calculation of $\hat{G}_D$ (whereas the naive gradient update only uses data from the current round). Second, we expect the demand function $D(r)$ to be simpler and more nicely structured than the full revenue function $\mu(r)$---for example, $D(r)$ is weakly decreasing in $r$---and therefore more amenable to parametric modeling.

\section{Omitted Proofs}

\subsection{Formal convergence rate}\label{app:convergence-rate-proof}

To show a convergence result for a non-convex problem with constraints, a measure called \emph{gradient mapping} is widely used in the literature~e.g.~\cite{GLZ2013, Reddi2016ProximalSM, LLCHA2018}. We define the gradient mapping used in this paper as follows,

\begin{definition}[Gradient Mapping]\label{def:gradient_mapping}
Let $f(x)$ be a differentiable function defined on $[0,1]$, $\mathcal{C} \in (0,1)$ is a convex space, and $\PiC$ is the projection operator defined as

\begin{equation}\label{eq:projection-operator}
\PiC(x) = \argmin_{z\in\mathcal{C}}|z-x|, \forall x\in \R.
\end{equation}

The gradient mapping $\PC$ is then defined as 
\begin{equation*}
\PC(r, \hat{g}, \alpha) := \frac{1}{\alpha}(\PiC(r + \alpha \hat{g}) - r)
\end{equation*}
where $\hat{g}$ is a gradient estimate of $\mu(r)$ (can be biased), $r\in \mathcal{C}$ is the reserve and $\alpha$ is the learning rate.
\end{definition}

The gradient mapping $\PC$ can be interpreted as the \emph{projected
gradient}, which offers a feasible update from the previous reserve $r$. Indeed, if the projection operator is the identity function, the gradient mapping just returns the gradient.

\begin{theorem} \label{thm:convergence-rate}
Suppose $\mu(r)$ is $L$-smooth. Let $\hat{G}_t =\frac{1}{n_t}\sum_{i=1}^{n_t}\hat{G}_{i,t}$ be the gradient estimator at time $t$, where $\forall i\in[n_t], |\hat{G}_{i,t}|\leq C$ almost surely. Fix $\alpha_t = \Theta(\frac{1}{\sqrt{T}})$, $\beta_t = \delta/2r_t$ and $n_t = N/T$. Assume for all $t$ that $\Bias(\hat{G}_t) \leq B$ and $\Var(\hat{G}_t) \leq V$. Then with probability at least $1-\delta$, Algorithm~\ref{alg:OSGD-revenue-general} satisfies that
{\small
\begin{eqnarray*}
&&\min_{t\in[T]} |\PC(r_t, \mu'(r_t), \alpha_t)|^2 \\
&\leq& \mathcal{O}\left(\sqrt{\frac{1}{T}} + \left(\delta + B + \frac{CT\ln(2T/\delta)}{N} + \sqrt{V\ln(2T/\delta)}\right)^2\right).
\end{eqnarray*}}
where $\PC$ is the gradient mapping.
\end{theorem}
To prove Theorem~\ref{thm:convergence-rate}, we first show some useful inequalities summarized in Lemma~\ref{lem:bernstein-ie} and Lemma~\ref{lem:projected-gradient-bound}.

\begin{lemma}[Bernstein Inequality]\label{lem:bernstein-ie}
Let $G = \frac{1}{n}\sum_{i=1}^n G_i$ be the random variable of estimation of revenue's gradient ($G_i$ can be correlated), where $|G_i| \leq C$ almost surely, and $V=\Var[G]$. Then, we have
\begin{align*}
&\left\vert\frac{1}{n}\sum_{i=1}^{n} G_{i} - \E\left[\frac{1}{n}\sum_{i=1}^{n} G_{i}\right]\right\vert\leq \frac{2C\ln(2T/\delta)}{3n} + \sqrt{2V\ln(2T/\delta)}
\end{align*}
holds with probability at least $1-\frac{\delta}{T}$.
\end{lemma}

\begin{proof}
By Bernstein's inequality, we have
\begin{align*}
&\PP\left(\left\vert\frac{1}{n}\sum_{i=1}^{n}G_{i} - \frac{1}{n}\sum_{i=1}^{n} \E[G_{i}]\right\vert \leq z\right)\geq 1 - 2\exp\left(\frac{-z^2 n}{2(V+Cz/3n)}\right)
\end{align*}
Let $2\exp\left(\frac{-z^2}{2(V+Cz/3n)}\right) = \frac{\delta}{T}$ and solving for $z$, we have
\begin{align*}
z &= \frac{\frac{2C\ln(2T/\delta)}{3n} + \sqrt{\left(\frac{2C\ln(2T/\delta)}{3n}\right)^2 + 8V\ln(2T/\delta)}}{2}\leq \frac{2C\ln(2T/\delta)}{3n} + \sqrt{2V\ln(2T/\delta)} := z'
\end{align*}
where the last inequality is based on the fact $\sqrt{a+b}\leq \sqrt{a} + \sqrt{b}$. Therefore, we have
\begin{align*}
&\PP\left(\left\vert\frac{1}{n}\sum_{i=1}^{n} G_{i} - \frac{1}{n}\sum_{i=1}^{n} \E[G_{i}]\right\vert \leq z'\right)\geq \PP\left(\left\vert\frac{1}{n}\sum_{i=1}^{n} G_{i} - \frac{1}{n_t}\sum_{i=1}^{n} \E[G_{i}]\right\vert \leq z\right) \geq 1 - \frac{\delta}{T}.
\end{align*}
\end{proof}

\begin{lemma}\label{lem:projected-gradient-bound}
For any $t\in[T]$, we have
\begin{eqnarray*}
\hat{G}_t \cdot \PC(r_t, \hat{G}_t, \alpha_t) &\geq& |\PC(r_t, \hat{G}_t, \alpha_t)|^2\\
|\PC(r_t, \mu'(r_t), \alpha_t) - \PC(r_t, \hat{G}_t, \alpha_t)| &\leq& |\mu'(r_t) - \hat{G}_t|
\end{eqnarray*}
\end{lemma}

\begin{proof}
Since $\mathcal{C}$ is a convex space, then for any $x\in \mathcal{C}, z\in \R$, we have $(x - \PiC(z))\cdot(z - \PiC(z)) \leq 0$. Let $z = r_t + \alpha_t \hat{G}_t$ and $x = r_t$, we have 
\begin{eqnarray*}
(r_t - \PiC(r_t + \alpha_t \hat{G}_t)) \cdot (r_t + \alpha_t \hat{G}_t - \PiC(r_t + \alpha_t \hat{G}_t)) \leq 0,
\end{eqnarray*}
which implies
\begin{eqnarray*}
\alpha_t \hat{G}_t (\PiC(r_t + \alpha_t \hat{G}_t)- r_t) \geq |r_t - \PiC(r_t + \alpha_t \hat{G}_t)|^2,
\end{eqnarray*}
Thus, we have $\hat{G}_t \cdot \PC(r_t, \hat{G}_t, \alpha_t) \geq |\PC(r_t, \hat{G}_t, \alpha_t)|^2$. Again, since $\mathcal{C}$ is a convex space, $\forall x, z \in \R, |\PiC(x) - \PiC(z)|\leq |x - Z|$. Then we can prove the second inequality,
\begin{eqnarray*}
|\PC(r_t, \mu'(r_t), \alpha_t) - \PC(r_t, \hat{G}_t, \alpha_t)|\leq \frac{1}{\alpha_t}|\PiC(r_t + \alpha_t \mu'(r_t)) - \PiC(r_t + \alpha_t \hat{G}_t|\leq |\mu'(r_t) - \hat{G}_t|
\end{eqnarray*}

\end{proof}

\begin{proof}[Proof of Theorem~\ref{thm:convergence-rate}]
Denote $r_t^+ = (1 + \beta_t)r_t$ and $r_t^- = (1-\beta_t)r_t$.
First we bound the bias of $\hat{G}_t$ compared with $\mu'(r_t)$. $|\mu'(r_t) - \hat{G}_t|$ can be decomposed as follows, 
\begin{equation}\label{eq:gradient-estimation-error1}
\begin{aligned}
\left|\mu'(r_t) - \hat{G}_t\right| &\leq \left|\mu'(r_t) - \frac{\mu(r_t^+) - \mu(r_t^-)}{r_t^+ - r_t^-}\right| + \left|\E[\hat{G}_t] - \frac{\mu(r_t^+) - \mu(r_t^-)}{r_t^+ - r_t^-}\right|+ \left|\E[\hat{G}_t] - \hat{G}_t\right| 
\end{aligned}
\end{equation}

Then we bound the three terms above separately. For the first term, we have
\begin{align*}
&\mu'(r_t) - \frac{\mu(r_t^+) - \mu(r_t^-)}{r_t^+ - r_t^-}= \frac{\mu'(r_t)\beta_t r_t -  \mu(r_t^+) + \mu'(r_t)\beta_t r_t + \mu(r_t^-)}{r_t^+ - r_t^-}
\end{align*}

By the smoothness of $\mu(r)$, 
\begin{align*}
\left|\mu(r_t^+) - \mu(r_t) - \mu'(r_t)\cdot\beta_t r_t \right|\leq \frac{L}{2}\cdot\beta_t^2 r_t^2\\
\left|\mu(r_t^-) - \mu(r_t) + \mu'(r_t)\cdot\beta_t r_t\right| \leq \frac{L}{2}\cdot\beta_t^2 r_t^2,
\end{align*}

Thus, we get
\begin{align*}
\left|\mu'(r_t) - \frac{\mu(r_t^+) - \mu(r_t^-)}{r_t^+ - r_t^-}\right| \leq \frac{1}{2\beta_t r_t} \cdot L \beta_t^2 r^2_{t}= \frac{L \beta_t r_t}{2}
\end{align*}

By assumption, the second term is bounded by $B$. Combining Lemma~\ref{lem:bernstein-ie} in Appendix, for any fixed $t\in[T]$,
\begin{eqnarray}\label{eq:gradient-estimation-error2}
\left|\mu'(r_t) - \hat{G}_t\right| \leq \frac{L\beta_t r_t}{2} + B + \frac{2C\ln(2T/\delta)}{3n_t} + \sqrt{V\ln(2T/\delta)} = B_{\delta, t},
\end{eqnarray}
holds with probability at least $1-\frac{\delta}{T}$. 

The $L$-smoothness of revenue function $\mu(r)$ implies the following inequalities, for any $t=1,\cdots, T$,

\begin{eqnarray*}
& &\mu(r_{t+1})\\
&\geq& \mu(r_t) + \mu'(r_t)\cdot(r_{t+1} - r_t) - \frac{L}{2}(r_{t+1} - r_t)^2\\
&\geq &\mu(r_t) +  \alpha_t\mu'(r_t)\cdot\PC(r_t, \hat{G}_t, \alpha_t) - \frac{L\alpha_t^2}{2}|\hat{G_t}|^2\\
&=& \mu(r_t) +  \alpha_t \hat{G}_t\cdot\PC(r_t, \hat{G}_t, \alpha_t) + \alpha_t (\mu'(r_t) - \hat{G}_t)\cdot\PC(r_t, \hat{G}_t, \alpha_t)- \frac{L\alpha_t^2}{2}|\hat{G_t}|^2\\
& \geq & \mu(r_t) +  \alpha_t |\PC(r_t, \hat{G}_t, \alpha_t)|^2 + \alpha_t (\mu'(r_t) - \hat{G}_t)\cdot\PC(r_t, \hat{G}_t, \alpha_t) - \frac{L\alpha_t^2}{2}|\hat{G_t}|^2
\end{eqnarray*}

The first inequality is because $|r_{t+1} - r_t| \leq \alpha_t \hat{G}_t$ and the last inequality is based on first inequality in Lemma~\ref{lem:projected-gradient-bound}. Rearranging the above inequalities, we have,
\begin{equation}\label{eq:stationary-ub-1}
\begin{aligned}
\alpha_t |\PC(r_t, \hat{G}_t, \alpha_t)|^2 &\leq  \mu(r_{t+1}) - \mu(r_t) +\alpha_t (\hat{G}_t - \mu'(r_t))\cdot\PC(r_t, \hat{G}_t, \alpha_t) + \frac{L\alpha_t^2}{2}|\hat{G}_t|^2\\
&\leq \mu(r_{t+1}) - \mu(r_t) + \frac{\alpha_t |\hat{G}_t - \mu'(r_t)|^2}{2} + \frac{\alpha_t|\PC(r_t, \hat{G}_t, \alpha_t)|^2}{2} + \frac{L\alpha_t^2}{2}|\hat{G}_t|^2
\end{aligned}
\end{equation}

The last inequality is based on Cauchy-Schwartz inequality and the second statement in Lemma~\ref{lem:projected-gradient-bound}. Rearranging Equation~(\ref{eq:stationary-ub-1}), we have %
\begin{equation}\label{eq:stationary-ub-2}
\begin{aligned}
\alpha_t|\PC(r_t, \hat{G}_t, \alpha_t)|^2 \leq 2(\mu(r_{t+1}) - \mu(r_t))+ \alpha_t|\hat{G}_t - \mu'(r_t)|^2 + L\alpha^2_t|\hat{G}_t|^2
\end{aligned}
\end{equation}

Using Cauchy-Schwartz inequality and Lemma~\ref{lem:projected-gradient-bound}, we can bound $|\PC(r_t, \mu'(r_t), \alpha_t)|$ in the following way,
\begin{eqnarray*}
|\PC(r_t, \mu'(r_t), \alpha_t)|^2 &\leq& 2|\PC(r_t, \hat{G}_t, \alpha_t)|^2 + 2|\PC(r_t, \mu'(r_t), \alpha_t) - \PC(r_t, \hat{G}_t, \alpha_t)|^2\\
&\leq & 2|\PC(r_t, \hat{G}_t, \alpha_t)|^2 + 2 |\mu'(r_t) - \hat{G}_t|^2
\end{eqnarray*}

Then we can write down the following bound,
\begin{equation}\label{eq:stationary-ub-3}
\begin{aligned}
\alpha_t|\PC(r_t, \mu'(r_t), \alpha_t)|^2 & \leq 2\alpha_t|\PC(r_t, \hat{G}_t, \alpha_t)|^2 + 2\alpha_t|\mu'(r_t) - \hat{G}_t|^2\\
& \leq 4(\mu(r_{t+1}) - \mu(r_t))+ 4\alpha_t|\hat{G}_t - \mu'(r_t)|^2 + 2L\alpha^2_t|\hat{G}_t|^2\\
& \leq 4(\mu(r_{t+1}) - \mu(r_t))+ 4\alpha_t |\hat{G}_t - \mu'(r_t)|^2 + 4L\alpha_t^2(|\hat{G}_t - \mu'(r_t)|^2 + |\mu'(r_{t})|^2)
\end{aligned}
\end{equation}

Let $r^*\in \mathcal{C}$ be the point with the minimum absolute gradient, i.e., $r^* = \argmin_{r\in\mathcal{C}}|\mu'(r)|$.
Notice $|\mu'(r_t)|^2 \leq 2|\mu'(r^*) |^2 + 2|\mu'(r^*) - \mu'(r_t)| \leq 2L + 2|\mu'(r^*)|^2 := L^*$. Summing up the inequality~(\ref{eq:stationary-ub-3}) from $t=1$ to $T$, based on inequality~(\ref{eq:gradient-estimation-error2}) with probability at least $1-\delta$, we have
\begin{eqnarray*}
\sum_{t=1}^T \alpha_t|\PC(r_t, \mu'(r_t), \alpha_t)|^2\leq 4(\mu(r_{T+1}) - \mu(r_1)) + 4\sum_{t=1}^T \alpha_t |B_{\delta, t}|^2 + 4L\sum_{t=1}^T \alpha_t^2 (|B_{\delta, t}|^2 + L^*)
\end{eqnarray*}

Setting $\alpha_t = \Theta\left(\frac{1}{\sqrt{T}}\right)$, $\beta_t = \delta/2r_t$ and $n_t = N/T$, By Cauchy-Schwartz inequality and the fact that $\sum_{t=1}^T \alpha_t = \Theta(\sqrt{T})$ and $\sum_{t=1}^T \alpha_t^2 = \Theta(1)$.  Therefore, we get 
\begin{eqnarray*}
\min_{t=1,\cdots,T}|\PC(r_t, \mu'(r_t), \alpha_t)|^2&\leq& \frac{1}{\sum_{t=1}^T \alpha_t} \sum_{t=1}^{T}\alpha_t |\PC(r_t, \mu'(r_t), \alpha_t)|^2\\
&\leq &\mathcal{O}\left(\sqrt{\frac{1}{T}} + \left(\delta + B + \frac{CT\ln(2T/\delta)}{N} + \sqrt{V\ln(2T/\delta)}\right)^2\right)
\end{eqnarray*}
\end{proof}

\subsection{Proof of Theorem~\ref{thm:discrete-gradient-revenue}}
\begin{proof}
Since $X^+_i$ and $X^-_i$ are independent random samples from $\B_{r^+}$ and $\B_{r^-}$ respectively, $\E[\hat{G}] = \frac{\mu(r^+) - \mu(r^-)}{r^+ - r^-}$. For the variance, since $X^+_i$ is bounded by $[0,1]$, then the variance of each $X^+_i$ and $X^-_i$ is at most $1/4$, which implies,
\begin{eqnarray*}
\Var(\hat{G}) \leq \frac{2}{4\delta^2 n} = \frac{1}{2\delta^2 n}
\end{eqnarray*}
\end{proof}

\subsection{Variance bounds for \texorpdfstring{$\hat{G}_D$}{Lg}}

In this subsection we bound the variance of the unbiased demand estimator

$$\hat{G}_D = \frac{r^{+}\sum_{i=1}^{n}\1(X^{+}_i \geq r^{+}) - r^{+}\sum_{i=1}^{n}\1(X^{+}_i \geq r^{+})}{n(r^{+} - r^{-})}.$$

Since $\E\left[\frac{1}{n}\sum_{i=1}^{n}\1(X^{+}_i \geq r^{+})\right] = D(r^{+})$ and $\E\left[\frac{1}{n}\sum_{i=1}^{n}\1(X^{-}_i \geq r^{-})\right] = D(r^{-})$, it follows immediately that

$$\E[\hat{G}_D] = \frac{r^{+}D(r^{+}) - r^{-}D(r^{-})}{r^{+} - r^{-}}$$

and therefore that $\Bias(\hat{G}_D) = 0$. In the following theorem, we bound $\Var(\hat{G}_D)$.

\begin{theorem}\label{thm:var_demand}
Let $\delta = r^{+} - r^{-}$. We have that

$$\Var(\hat{G}_D) \leq \frac{(r^{+})^2}{2n\delta^2}.$$
\end{theorem}
\begin{proof}
Since the $X^{+}_i$ and $X^{-}_i$ are independent random variables, their variances are additive so

\begin{eqnarray*}
\Var(\hat{G}_D) &=& \frac{1}{n^2\delta^2}\left(n(r^{+})^2\Var(\1(X^{+}_i \geq r^{+})) + n(r^{-})^2\Var(\1(X^{-}_i \geq r^{-}))\right) \\
&\leq& \frac{(r^{+})^2}{2n\delta^2},
\end{eqnarray*}

where in the last line we have used the fact that the variance of a Bernoulli random variable is bounded above by $1/4$.
\end{proof}

\subsection{Proof of Theorem~\ref{thm:revenue-decomposition}}
\begin{proof}
Note that

\begin{eqnarray*}
\mu(r) &=& \E_{v\sim \F}[b(r, v)] \\
&=& \E_{v\sim \F}\left[\max(b(r, v), r) - r\1\{b(r, v) = 0\}\right] \\
&=& \E_{v\sim \F}\left[\max(b(r, v), r)\right] - r\E_{v\sim \F}[\1\{b(r, v) = 0\}] \\
&=& \E_{v\sim \F}\left[\max(b(r, v), r)\right] - r\Pr_{v \sim \F}[v < r] \\
&=& \E_{v\sim \F}\left[\max(b(r, v), r)\right] - r(1 - \Pr_{v \sim \F}[v \geq r]) \\
&=& \E_{v\sim \F}\left[\max(b(r, v), r) - r\right] + r\Pr_{v \sim \F}[v \geq r])
\end{eqnarray*}
\end{proof}

\subsection{Proof of Theorem~\ref{thm:bias-variance-eps-bounded}}\label{app:eps-bounded-response}
In the following proof we slightly abuse notation, by denoting $\B_r$ to be the CDF of the bid distribution w.r.t reserve price $r$. We also formally state the definition of $\eps$-bounded response model here.

\begin{definition}[$\eps$-bounded response model]\label{def:eps-bounded-response-model}
A $\eps$-bounded response bidding function takes the form,
$$
b(r, v) = \left\{
\begin{array}{cc}
b(0, v) & \text{ if } b(0, v) \geq r  \\
r + z & \text{ if } b(0, v) \leq r \leq v\\
0  & \text{ if } v < r \\
\end{array}
\right.
$$
where $z \in [0, \eps]$, and $z$ can be a random variable.
\end{definition}

\begin{proof}
To bound the variance, note that each $Y^{-}_i$ is constrained to the interval $[0, r^{+} - r^{-}]$. Since this interval has length at most $\delta$, the variance of each $Y^{-}_i$ is at most $(\delta^2/4)$, then 

$$\Var(\hat{G}_E) \leq \frac{n(\delta^2/4)}{n^2\delta^2} \leq \frac{1}{4n}.$$

Then we focus on bounding the bias, 

\begin{eqnarray*}
& & \E_{v}\left[\max(b(r^{+}, v) - r^+, 0) - \max(b(r^{-}, v) - r^-, 0)\right]\\
&=& \int_{0}^{1} \max(b, r^{+}) d\B_{r^+} - \int_0^1 \max(b, r^{-}))d\B_{r^-} - (r^+ - r^-) \\
&=& \int_{0}^{r^{+}+\eps} \max(b, r^{+}) d\B_{r^+} - \int_0^{r^++\eps} \max(b, r^{-})d\B_{r^-}  + \int_{r^{+}+\eps}^{1} \max(b, r^{+}) d\B_{r^+} - \int_{r^{+}+\eps}^{1} \max(b, r^{-}) d\B_{r^-}\\
& & - (r^+ - r^-) \\
&=& \int_{0}^{r^{+}} \max(b, r^{+})d\B_{r^+} -\int_0^{r^+} \max(b, r^{-})d\B_{r^-} + \int_{r^{+}}^{r^{+}+\eps} \max(b, r^{+}) d\B_{r^+} - \int_{r^{+}}^{r^{+}+\eps} \max(b, r^{-})d\B_{r^-}\\
&& - (r^+ - r^-)\\
&=& r^+ \B_{r^+}(r^+) - \int_0^{r^+} \max(b, r^{-})d\B_{r^-} - (r^+ - r^-) + \int_{r^{+}}^{r^{+}+\eps} b d\B_{r^+} - \int_{r^{+}}^{r^{+}+\eps} b d\B_{r^-}
\end{eqnarray*}

Here the third equality holds because the fact that if $b > r^+ +\eps, \B_{r^+}(b) = \B_{r^-}(b)$ by property of the $\eps$-bounded response. The fourth equality is based on $\max(b(r^+, v), r^+) = r^+$ when $v \leq r^+$ and $\F(r^+) = \B_{r^+}(r^+)$. Then we consider $\E[Y^-_i]$, where

\begin{eqnarray*}
\E[Y^-_i] &=& \E[\max(X^-_i - r^-, 0)\1\{X^-_i \leq r^+\}] + \E[(r^+-r^-)\1\{X^-_i > r^+\}]\\
& = & \int_0^{r^+} \max(b, r^-) d\B_{r^-} + r^+ (1 - \B_{r^-}(r^+)) - r^-
\end{eqnarray*}

Before bounding the bias of $\hat{G}_E$, we state some useful equations based on integral by part.
\begin{eqnarray}
\int_{r^{+}}^{r^{+}+\eps} b d\B_{r^+} &=& (r^+ + \eps)\B_{r^+}(r^+ + \eps) - r^+ \B_{r^+}(r^+) - \int_{r^{+}}^{r^{+}+\eps}\B_{r^+}(b) db\label{eq:integral-1}\\
\int_{r^{+}}^{r^{+}+\eps} b d\B_{r^-} &=& (r^+ + \eps)\B_{r^-}(r^+ + \eps) - r^+ \B_{r^-}(r^+) - \int_{r^{+}}^{r^{+}+\eps}\B_{r^-}(b) db\label{eq:integral-2}
\end{eqnarray}

Based on definition of $\eps$-bounded response, $\B_{r^-}(r^+ + \eps) = \B_{r^+}(r^+ + \eps)$. Then we have

\begin{eqnarray*}
&&\left|\E[\hat{G}_E] - \frac{E(r^{+}) - E(r^{-})}{r^{+} - r^{-}} \right| \\
&=& \left|-\frac{\E[Y^-_i]}{r^+ - r^-} - \frac{E(r^{+}) - E(r^{-})}{r^{+} - r^{-}} \right|\\
&=& \frac{1}{r^{+} - r^{-}} \left|r^+ \B_{r^+}(r^+) - r^+\B_{r^-}(r^+) + \int_{r^{+}}^{r^{+}+\eps} b d\B_{r^+} - \int_{r^{+}}^{r^{+}+\eps} b d\B_{r^-}\right|\\
& = & \frac{1}{r^{+} - r^{-}} \left|\int_{r^{+}}^{r^{+}+\eps}\B_{r^+}(b) db - \int_{r^{+}}^{r^{+}+\eps}\B_{r^-}(b) db \right| \\
&&(\text{Based on Equations~(\ref{eq:integral-1}) and~(\ref{eq:integral-2}) as well as }\B_{r^-}(r^+ + \eps) = \B_{r^+}(r^+ + \eps))\\
&\leq &\frac{2\eps}{r^{+} - r^{-}}\\
\end{eqnarray*}
where the inequality is because $\forall b\in[r^{+}, r^{+}+\eps], \B_{r^+}(b), \B_{r^-}(b) \leq 1$.
\end{proof}

\subsection{Proof of Theorem \ref{thm:quantile-vr}}\label{app:quantile-vr-proof}

We start with the following helpful auxiliary lemmas.

\begin{lemma}\label{lem:order-statistic}
Let $Y_1, Y_2, \dots, Y_{n}$ be $n$ iid uniform random variables. Let $Y^{(k)}$ be the $k$th largest $Y_i$. Then with probability at least $1-n^{-2/3}$, 

$$\left|Y^{(k)} - \frac{k}{n+1}\right| \leq n^{-2/3}.$$
\end{lemma}
\begin{proof}
From the theory of order statistics \cite{gentle2009computational}, we know that $Y^{(k)} \sim \mathrm{Beta}(k, n+1 - k)$. It is known that $\E[Y^{(k)}] = \frac{k}{n+1}$ and that $\Var(Y^{k}) \leq 1/(8n)$. The statement immediately follows from Chebyshev's inequality. 
\end{proof}

\begin{lemma}\label{lem:order-bias}
Let $f:[0, 1] \rightarrow [0, 1]$ be an increasing function and let $Y_1, Y_2, \dots, Y_n$ be $n$ iid uniform random variables. Let $X_i = f(Y_i)$, and let $S_k$ be the r.v. equal to the sum of the $k$ smallest $X_i$. Then 

$$\left|\frac{1}{n}\E[S_{qn}] - \int_{0}^{q}f(x)dx\right| \leq 3n^{-2/3}$$
\end{lemma}
\begin{proof}
Let $Z_1, Z_2, \dots, Z_{qn}$ be (a random permutation) of the $qn$ smallest $Y_i$ (so $S_k = \sum f(Z_{i})$). Note that conditioned on $Z_{qn+1} = r$, the $Z_i$ are independently distributed according to $U([0, r])$. In particular, we have that

$$\frac{1}{n}\E[S_{qn} | Z_{qn+1} = r] = \int_{0}^{r}f(x)dx.$$

From Lemma \ref{lem:order-statistic}, we know that with probability at least $1 - n^{-2/3}$, $Z_{qn+1} \in [q - 2n^{-2/3}, q + 2n^{-2/3}]$. Since $f(x) \in [0, 1]$, it follows that $\int_{0}^{r}f(x)dx$ is 1-Lipshitz and therefore (conditioned on $Z_{qn+1} \in [q - 2n^{-2/3}, q + 2n^{-2/3}]$),

$$\left|\int_{0}^{q}f(x)dx - \int_{0}^{Z_{qn+1}}f(x)dx\right| \leq 2n^{-2/3}.$$

On the other hand, in the ($n^{-2/3}$ probability) case where $Z_{qn+1}\notin [q - 2n^{-2/3}, q + 2n^{-2/3}]$, $\frac{1}{n}\E[S_{qn}]$ is still bounded in $[0, 1]$. The theorem statement immediately follows. 
\end{proof}

\begin{lemma}\label{lem:efron-stein-order}
Let $X_i$ be an iid collection of $n$ rvs. Let $X^{(k)}$ be the $k$th smallest of the $X_i$ (so $X^{(1)} \leq X^{(2)} \leq \dots \leq X^{(n)}$). Then, if $S_k = \sum_{i=1}^{k} X^{(i)}$, we have that

$$\Var(S_k) \leq n \E[(X^{(k)})^2].$$
\end{lemma}
\begin{proof}
The Efron-Stein inequalities (see Theorem 2 of \cite{boucheron2012concentration}) state that for any collection of $n$ random variables $X_i$ and any measurable functions $f: \R^n \rightarrow \R$ and $f_i: \R^{n-1} \rightarrow \R$ we have that

$$\Var(f(X)) \leq \sum_{i=1}^{n}\E[(f(X) - f_i(X_{-i}))^2],$$

where $X_{-i}$ is the $(n-1)$-tuple of rvs $(X_1, X_2, \dots, X_{i-1}, X_{i+1}, \dots, X_{n})$. 

Let $f(X)$ equal the sum of the $k$ smallest entries in $X$, and let $f_i(X')$ equal the sum of the $k-1$ smallest entries in $X'$. Note that for this choice of $f$ and $f_i$, $f(X) = S_k$, and $0 \leq f(X) - f_i(X_{-i}) \leq X^{(k)}$ (since the $k-1$ smallest entries in $X_{-i}$ are a subset of the $k$ smallest entries in $X$). It follows that $\Var(S_k) \leq n \E[(X^{(k)})^2]$, as desired.
\end{proof}

We can now proceed to prove Theorem \ref{thm:quantile-vr}.

\begin{proof}[Proof of Theorem \ref{thm:quantile-vr}]
We begin by bounding the variance of our estimator $\hat{G}_E$. Let us begin by focusing on $\Var(\sum_{i=1}^{qn}Y_{i}^{-})$. Since these $Y_i$ are sorted, Lemma \ref{lem:efron-stein-order} implies that $\Var(\sum_{i=1}^{qn}Y_{i}^{-}) \leq n \E[(Y_{qn}^{-})^2]$. Since $\tilde{t} = \F^{-1}(q + n^{-2/3})$, by Lemma \ref{lem:order-statistic}, with probability at least $1 - n^{-2/3}$, $Y_{qn}^{-} \leq \tilde{t}$, and therefore $\Var(\sum_{i=1}^{qn}Y_{i}^{-}) \leq n \tilde{t}^2 + n^{1/3}$. Similarly, $\Var(\sum_{i=1}^{qn}Y_{i}^{+}) \leq n \tilde{t}^2 + n^{1/3}$. Since the sets of rvs $Y^{-}_i$ and $Y^{+}_i$ are independent, we have that

$$\Var(\hat{G}_{E}) \leq \frac{2\tilde{t}^2}{n\delta^2} + O(n^{-5/3}\delta^{-2}).$$

We now proceed to bound the bias of $\hat{G}_E$. First, note that by Lemma \ref{lem:order-bias}, we have that

$$\left|\frac{1}{n}\E\left[\sum_{i=1}^{qn}Y^{-}_{qn}\right] - \int_{0}^{q}\max(b(r^{-}, \F^{-1}(x)) -r^{-}, 0)dx\right| \leq 3n^{-2/3},$$

\noindent
and therefore

$$\left|\frac{1}{n}\E\left[\sum_{i=1}^{qn}Y^{-}_{qn}\right] - \int_{0}^{t}\max(b(r^{-}, v) -r^{-}, 0)d\F(v)\right| \leq 3n^{-2/3}.$$

Likewise

$$\left|\frac{1}{n}\E\left[\sum_{i=1}^{qn}Y^{+}_{qn}\right] - \int_{0}^{t}\max(b(r^{+}, v) -r^{+}, 0)d\F(v)\right| \leq 3n^{-2/3}.$$

It follows that 

\begin{equation}\label{eq:half-bias}
\left|\hat{G}_E - \left(\frac{1}{\delta}\int_{0}^{t}(\max(b(r^{+}, v) -r^{+}, 0) - \max(b(r^{-}, v) - r^{-}, 0)) d\F(v) - (1-q)\right)\right| \leq 6n^{-2/3}.
\end{equation}

On the other hand, note that

\begin{eqnarray*}
E(r^{+}) - E(r^{-}) &=& \int_{0}^{1}(\max(b(r^{+}, v) -r^{+}, 0) - \max(b(r^{-}, v) - r^{-}, 0)) d\F(v) \\
&=& \int_{0}^{t}(\max(b(r^{+}, v) -r^{+}, 0) - \max(b(r^{-}, v) - r^{-}, 0)) d\F(v) \\
&&\,+ \int_{t}^{1}(\max(b(r^{+}, v) -r^{+}, 0) - \max(b(r^{-}, v) - r^{-}, 0)) d\F(v).
\end{eqnarray*}

Now, note that for $v \geq r^{+}$, $\max(b(r^{+}, v) -r^{+}, 0) - \max(b(r^{-}, v) - r^{-}, 0) = b(r^{+}, v) - b(r^{-}, v) - (r^{+} - r^{-})$. Since $\Pr[v \geq t] = 1-q$, it follows that

\begin{eqnarray*}
& & \int_{t}^{1}(\max(b(r^{+}, v) -r^{+}, 0) - \max(b(r^{-}, v) - r^{-}, 0)) d\F(v)\\
&=& \int_{t}^{1} (b(r^{+}, v) - b(r^{-}, v) - (r^{+} - r^{-})) d\F(v)\\
&=& \int_{t}^{1} (b(r^{+}, v) - b(r^{-}, v)) d\F(v) - (1-q)(r^{+} - r^{-}) \\
&\in & \left[-(1-q)\delta, -(1-q)\delta + (b(r^{+}, t) - b(r^{-}, t))(1-q)\right].
\end{eqnarray*}

Here the last line follows since $b(r^{+}, v) - b(r^{-}, v)$ is decreasing in $v$ (due to diminishing sensitivity to reserve) but always non-negative. Combining this with equation \ref{eq:half-bias}, we have that:

$$\left|\hat{G}_E - \frac{1}{\delta}(E(r^{+}) - E(r^{-}))\right| \leq \frac{(b(r^{+}, t) - b(r^{-}, t))(1-q)}{\delta} + 6n^{-2/3},$$

as desired.

\end{proof}

\subsection{Convergence Rate of Quantile Truncation}\label{app:convergence-rate-quantile-truncation}
\begin{corollary}\label{cor:convergence-rate-quantile-truncation}
Suppose $\Bias(\hat{G}_D) \leq \eps_D/\delta$. Using the estimator $\hat{G}_E$ proposed in Eq.~(\ref{eq:discrete-gradient-bid-truncation}) for the response model with diminishing sensitivity property, for any fixed quantile $q$, setting $T = N^{2/3}$ and $\delta = \Theta(\sqrt{\eps_D + 1 - q})$, Algorithm \ref{alg:OSGD-revenue-general} achieves convergence, $$\min_{t\in[T]} |\PC^t|^2 \leq \widetilde{\mathcal{O}}\left(\eps_D+ 1- q + \Big(1 + \frac{\mathcal{F}^{-1}(q + N^{-2/9})}{\eps_D + 1 - q}\Big) \cdot N^{-1/3}\right)$$
\end{corollary}

\section{Additional Experiments}\label{app:additional-experiment}
In this section, we show the parameters used in the algorithms and some additional experiments. For additional experiments, we compare the two truncation methods for perfect response models, and then we show the complete results of 20 semi-synthetic data sets with perfect response. Finally, we test for other different response models in synthetic data and one of the semi-synthetic data sets (the first data set). %
In addition to the figures about the revenue curve learned by our algorithms, we also report the average revenue of the first several rounds learned by our algorithms in Tables (see Table~\ref{tab:synthetic-revenue-50},~\ref{tab:semi-synthetic-revenue-20}, and~\ref{tab:semi-synthetic-revenue-50}).

{\bf Set up.} For all the algorithms, we set the learning rate to $0.05$, the minimum reserve price to $0.1$, the maximum reserve price to $5.0$, and the perturbation size to $\beta_t = 0.1$ at each round. For quantile truncation, we use the 80\% quantile as the threshold to discard bids. We run the algorithms for 200 rounds with access to 100 samples at each round (50 for $r^-$ and 50 for $r^+$), which forms 1 trial. We repeat 50 trials for each algorithm and report the mean (solid line) and 95\% confidence interval (translucent and colored area) of the revenue achieved during training for each algorithms shown in the figures. To obtain revenue curves learned by the algorithms over time, the revenue at each point is estimated through 10,000 bids randomly drawn from the bid distribution and bidder response model.

{\bf Comparison of bid truncation and quantile truncation for perfect response models.}  We show the average revenue achieved by algorithms with bid truncation and quantile truncation for synthetic data with perfect response in Table~\ref{tab:synthetic-revenue-50} (first 50th rounds) and semi-synthetic data with perfect response in Table~\ref{tab:semi-synthetic-revenue-20} (first 20th rounds) and Table~\ref{tab:semi-synthetic-revenue-50} (first 50th rounds). We show that the bid truncation (Algorithm (IV)) performs the best in most of cases, but quantile truncation (Algorithm (V)) is also competitive and performs well for perfect response in both synthetic data and semi-synthetic data sets. Bid truncation doesn't significantly outperform than quantile truncation significantly, since the algorithms only take small number of the samples at each round and the quantile truncation is more stable in this setting.

{\bf Performance of all 20 semi-synthetic data sets with perfect response.}
We evaluate the performance of five algorithms for 20 semi-synthetic data sets with perfect response. Here we still repeat 50 trials for each algorithm and report the mean and 95\% confidence interval of the revenue. The results are summarized in Table~\ref{tab:semi-synthetic-revenue-20} and Table~\ref{tab:semi-synthetic-revenue-50}, where each table records the average revenue over the first 20 rounds and 50 rounds, respectively. The revenue is normalized by the optimal revenue of each data set (empirically evaluated by grid search). We find in semi-synthetic data, the variance reduction methods (bid truncation and quantile truncation) improve the revenue achieved by the algorithms. Interestingly, we find Naive gradient descent with bid truncation also works well in several semi-synthetic data sets, this is because our demand modeling approach relies on a good estimator $\hat{G}_D$ and sometimes, the simple neural network cannot learn the demand curve very accurately in the beginning.

{\bf No response model. } In no response model, bidders don't change their bids based on reserve prices, i.e $\forall r > 0, r\geq 0, \mathcal{B}_r(r') = \mathcal{B}_0(r')$. This can be regarded as a perfect response model with linear shading factor 1. In this case, we expect the algorithms converge to the lower bound of the reserve prices 0.1. For no response models, we set perturbation be $0.3$ to speed up convergence and we still use logistic regression to learn demand curve. Figure~\ref{fig:no-response} shows that all the algorithms almost converge to minimum reserve price 0.1 and the bid truncation method works the best for no response model. 

{\bf Mixture of no response and perfect response models. }
We also consider another non-perfect response model, which is a mixture of perfect response and no response models.  We assume the bidder will not respond to reserve price with probability 0.1 and use perfect response with probability 0.9. Figure~\ref{fig:mixture-response} shows the revenue curve learned by the algorithms for synthetic data and one semi-synthetic data. We find quantile-based variance reduction speed up the training and converges to optimal reserve faster than naive gradient descent. Since this is a mixture of perfect response and no response model, the revenue achieved by bid truncation methods is worse than the quantile-based approach. Through this experiments, we find quantile truncation is not sensitive with different response models, whereas, the bid truncation method is very sensitive to a slightly non-perfect response model. 

{\bf $\eps$-bounded response. }
In the experiments for $\eps$-bounded response model, we set $\eps=0.05$ and the bias term $z\sim \mathtt{Unif}[0, \eps]$ in $\eps$-bounded response model (see definition~\ref{def:eps-bounded-response-model}). We visualize the revenue learned by algorithms for synthetic data and semi-synthetic data in Figure~\ref{fig:eps-bounded-response}. The figures show that for $\eps$-bounded response, the quantile truncation still works better than naive gradient descent. Since we have demonstrated that bid truncation is sensitive to non-perfect response model, the performance of the bid truncation is worse than quantile truncation in $\eps$-bounded response.

\begin{figure*}[t]
\begin{subfigure}[b]{0.5\textwidth}
\centering
\includegraphics[scale=0.4]{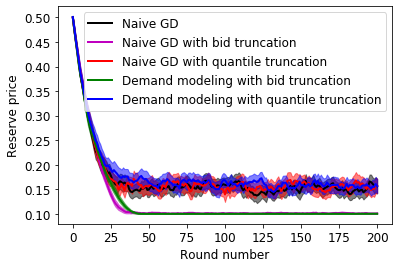}
\caption{Synthetic data with no response}
\label{fig:synthetic-no-response-revenue}
\end{subfigure}
\begin{subfigure}[b]{0.5\textwidth}
\centering
\includegraphics[scale=0.4]{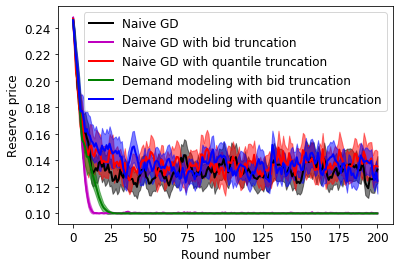}
\caption{Semi-synthetic data with no response}
\label{fig:dremel-no-response-revenue}
\end{subfigure}
\caption{Plots of reserve price as a function of round $t$ for (a) synthetic data with no response, and (b) one semi-synthetic data with no response.}
\label{fig:no-response}
\end{figure*}

\begin{figure*}[t]
\begin{subfigure}[b]{0.5\textwidth}
\centering
\includegraphics[scale=0.4]{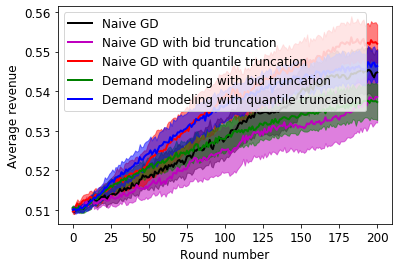}
\caption{Synthetic data with mixture response}
\label{fig:synthetic-mixture-response-revenue}
\end{subfigure}
\begin{subfigure}[b]{0.5\textwidth}
\centering
\includegraphics[scale=0.4]{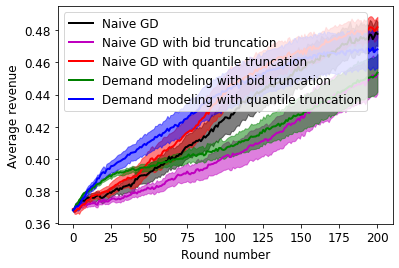}
\caption{Semi-synthetic data with mixture response}
\label{fig:semi-synthetic-mixture-response-revenue}
\end{subfigure}
\caption{Plots of reserve price and revenue as a function of round $t$ for (a) synthetic data and (b) one semi-synthetic data, with mixture response, where the bidder uses perfect response model with probability 0.9 and doesn't respond to the reserve, otherwise.}
\label{fig:mixture-response}
\end{figure*}

\begin{figure*}[t]
\begin{subfigure}[b]{0.5\textwidth}
\centering
\includegraphics[scale=0.4]{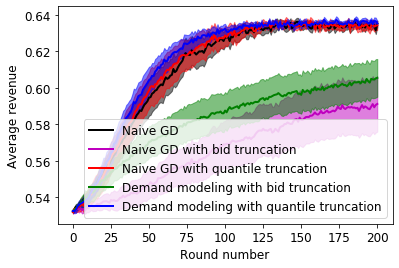}
\caption{Synthetic data with $\eps$-bounded response}
\label{fig:synthetic-eps-bounded-response}
\end{subfigure}
\begin{subfigure}[b]{0.5\textwidth}
\centering
\includegraphics[scale=0.4]{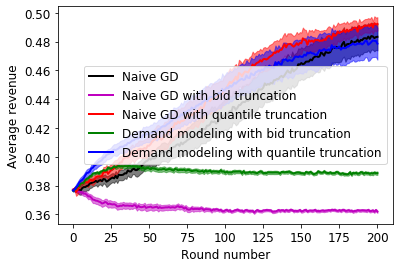}
\caption{Semi-synthetic data with $\eps$-bounded response}
\label{fig:semi-synthetic-eps-bounded-response}
\end{subfigure}
\caption{Plots of revenue as a function of round $t$ for (a) synthetic data with $\eps$-bounded response, and (b) one semi-synthetic data with $\eps$-bounded response, where $\eps=0.05$ and the bias term $z\sim \mathtt{Unif}[0, \eps]$.}
\label{fig:eps-bounded-response}
\end{figure*}

\begin{table}[ht]
\footnotesize
\centering
\begin{tabular}{|c|c|c|c|c|c|c|}
\hline
\multirow{2}{*}{Response model} & 
Algorithm (I) & Algorithm (II) & Algorithm (III) & Algorithm (IV) & Algorithm (V) 
\\
 & $\mathit{rev}$ & $\mathit{rev}$ & $\mathit{rev}$ & $\mathit{rev}$ & $\mathit{rev}$
\\
\hline
Perfect response&86.7$\pm$ 0.8\%&86.5$\pm$ 0.7\%&87.1$\pm$ 0.9\%&88.9$\pm$ 0.3\%& {\bf 89.4$\pm$ 0.6\%}
\\
\hline
Equilibrium response&95.5$\pm$ 0.5\%&90.1$\pm$ 0.2\%&96.3$\pm$ 0.4\%&93.5$\pm$ 0.6\%& {\bf 96.4$\pm$ 0.4\%}
\\
\hline
$\eps$-bounded response&86.4$\pm$ 0.8\%&83.0$\pm$ 0.6\%&86.6$\pm$ 1.0\%&85.1$\pm$ 0.8\%&{\bf 87.1$\pm$ 0.8\%}
\\
\hline
Mixture response &90.5$\pm$ 0.3\%&90.3$\pm$ 0.3\%&90.8$\pm$ 0.4\%&90.8$\pm$ 0.3\%&{\bf 91.0$\pm$ 0.4\%}
\\
\hline
\end{tabular}
\caption{Average revenue of the first 50th rounds of five algorithms for synthetic data with different response models. The revenue is normalized by the optimal revenue and we repeat 50 trials for each algorithm to report the 95\% confidence interval.
	\label{tab:synthetic-revenue-50}}
	\vspace*{-5pt}
\end{table}

\begin{table}[t]
\footnotesize
\centering
\begin{tabular}{|c|c|c|c|c|c|c|}
\hline
\multirow{2}{*}{Semi-synthetic data sets} & 
Algorithm (I) & Algorithm (II) & Algorithm (III) & Algorithm (IV) & Algorithm (V) 
\\
& $\mathit{rev}$ & $\mathit{rev}$ & $\mathit{rev}$ & $\mathit{rev}$ & $\mathit{rev}$
\\
\hline
(1)&73.9$\pm$ 0.4\%&73.8$\pm$ 0.2\%&74.1$\pm$ 0.4\%& {\bf 75.3$\pm$ 0.1\%} &75.3$\pm$ 0.3\%\\ \hline
(2)&91.8$\pm$ 0.8\%&92.9$\pm$ 0.3\%&92.9$\pm$ 0.6\%& {\bf 94.5$\pm$ 0.1\%} &94.3$\pm$ 0.6\%\\ \hline
(3)&94.1$\pm$ 0.2\%& {\bf 94.7$\pm$ 0.2\%} &93.5$\pm$ 0.3\%&94.4$\pm$ 0.1\%&93.3$\pm$ 0.3\%\\ \hline
(4)&82.9$\pm$ 0.4\%&83.4$\pm$ 0.0\%&82.9$\pm$ 0.3\%& {\bf 83.5$\pm$ 0.0\% }&83.2$\pm$ 0.2\%\\ \hline
(5)&93.0$\pm$ 0.6\%& {\bf 94.8$\pm$ 0.0\%} &93.4$\pm$ 0.4\%&94.4$\pm$ 0.0\%&93.9$\pm$ 0.4\%\\ \hline
(6)&93.8$\pm$ 0.7\%&91.8$\pm$ 0.5\%&94.0$\pm$ 0.6\%&95.6$\pm$ 0.1\%& {\bf 95.9$\pm$ 0.4\%}\\ \hline
(7)&86.4$\pm$ 0.4\%&85.1$\pm$ 0.4\%&87.4$\pm$ 0.5\%&87.9$\pm$ 0.1\%& {\bf 88.6$\pm$ 0.5\%} \\ \hline
(8)&95.0$\pm$ 0.7\%&96.0$\pm$ 0.2\%&95.3$\pm$ 0.6\%& {\bf 96.7$\pm$ 0.1\%} &95.6$\pm$ 0.3\%\\ \hline
(9)&83.6$\pm$ 1.0\%&84.0$\pm$ 0.5\%&84.9$\pm$ 0.9\%& {\bf 89.1$\pm$ 0.1\%} &88.6$\pm$ 1.1\%\\ \hline
(10)&94.3$\pm$ 0.8\%&94.8$\pm$ 0.3\%&94.9$\pm$ 0.8\%&{\bf 96.4$\pm$ 0.0\%} &95.9$\pm$ 0.5\%\\ \hline
(11)&52.0$\pm$ 1.1\%& {\bf 53.1$\pm$ 1.0\%} &52.2$\pm$ 1.6\%&52.1$\pm$ 1.0\%&51.2$\pm$ 1.6\%\\ \hline
(12)&86.6$\pm$ 0.7\%&86.9$\pm$ 0.3\%&87.9$\pm$ 0.7\%& {\bf 89.8$\pm$ 0.1\%} &89.6$\pm$ 0.5\%\\ \hline
(13)&89.7$\pm$ 1.5\%&89.8$\pm$ 0.8\%&89.1$\pm$ 1.5\%& {\bf 92.1$\pm$ 0.1\%} &91.2$\pm$ 1.0\%\\ \hline
(14)&93.3$\pm$ 0.6\%&93.0$\pm$ 0.4\%&93.1$\pm$ 0.7\%& {\bf 95.7$\pm$ 0.1\%} &95.2$\pm$ 0.7\%\\ \hline
(15)&96.2$\pm$ 0.7\%&95.0$\pm$ 0.5\%&96.3$\pm$ 0.7\%& {\bf 97.9$\pm$ 0.1\%} &97.8$\pm$ 0.3\%\\ \hline
(16)&88.4$\pm$ 0.8\%&87.8$\pm$ 0.4\%&88.6$\pm$ 0.8\%& {\bf 92.0$\pm$ 0.1\%} &92.1$\pm$ 0.7\%\\ \hline
(17)&94.6$\pm$ 0.6\%& {\bf 94.9$\pm$ 0.1\%} &93.9$\pm$ 0.7\%&92.5$\pm$ 0.1\%&91.6$\pm$ 0.6\%\\ \hline
(18)&91.0$\pm$ 0.6\%&90.9$\pm$ 0.3\%&90.8$\pm$ 0.6\%& {\bf 93.0$\pm$ 0.1\%} &92.7$\pm$ 0.4\%\\ \hline
(19)&83.4$\pm$ 0.6\%&83.4$\pm$ 0.4\%&84.7$\pm$ 0.6\%& {\bf 85.9$\pm$ 0.1\%} &85.8$\pm$ 0.6\%\\ \hline
(20)&94.7$\pm$ 0.9\%&{\bf 97.3$\pm$ 0.1\%} &95.6$\pm$ 0.7\%&93.9$\pm$ 0.1\%&92.2$\pm$ 0.4\%\\ \hline
\end{tabular}
\caption{Average revenue of first 20th rounds of five algorithms for semi-synthetic data sets with perfect response. The revenue is normalized by the optimal revenue of each data set. We repeat 50 trials to get the 95\% confidence interval.
	\label{tab:semi-synthetic-revenue-20}}
	\vspace*{-5pt}
\end{table}

\begin{table}[t]
\footnotesize
\centering
\begin{tabular}{|c|c|c|c|c|c|c|}
\hline
\multirow{2}{*}{Semi-synthetic data sets} & 
Algorithm (I) & Algorithm (II) & Algorithm (III) & Algorithm (IV) & Algorithm (V) 
\\
& $\mathit{rev}$ & $\mathit{rev}$ & $\mathit{rev}$ & $\mathit{rev}$ & $\mathit{rev}$
\\
\hline
(1)&75.1$\pm$ 0.6\%&74.5$\pm$ 0.3\%&75.7$\pm$ 0.6\%&76.8$\pm$ 0.2\%& {\bf 77.5$\pm$ 0.7\%}\\ \hline
(2)&93.4$\pm$ 0.5\%&94.1$\pm$ 0.2\%&94.4$\pm$ 0.4\%&94.8$\pm$ 0.1\%& {\bf 95.4$\pm$ 0.4\%} \\ \hline
(3)&93.9$\pm$ 0.2\%& {\bf 95.0$\pm$ 0.1\%} &93.3$\pm$ 0.2\%&94.9$\pm$ 0.1\%&93.4$\pm$ 0.3\%\\ \hline
(4)&83.3$\pm$ 0.3\%&83.7$\pm$ 0.0\%&83.3$\pm$ 0.3\%& {\bf 83.9$\pm$ 0.0\%} &83.7$\pm$ 0.1\%\\ \hline
(5)&95.7$\pm$ 0.3\%& {\bf 96.8$\pm$ 0.0\%} &95.9$\pm$ 0.2\%&96.4$\pm$ 0.0\%&96.3$\pm$ 0.2\%\\ \hline
(6)&95.2$\pm$ 0.4\%&93.7$\pm$ 0.5\%&95.4$\pm$ 0.4\%&{\bf 96.6$\pm$ 0.0\%} &96.6$\pm$ 0.2\%\\ \hline
(7)&88.8$\pm$ 0.5\%&87.1$\pm$ 0.5\%&90.0$\pm$ 0.6\%&90.0$\pm$ 0.1\%& {\bf 91.1$\pm$ 0.4\%}\\ \hline
(8)&97.3$\pm$ 0.3\%&98.1$\pm$ 0.1\%&97.0$\pm$ 0.4\%&{\bf 98.7$\pm$ 0.1\%}&97.4$\pm$ 0.2\%\\ \hline
(9)&91.5$\pm$ 0.8\%&91.9$\pm$ 0.6\%&92.8$\pm$ 0.8\%&{\bf 95.6$\pm$ 0.1\%}&95.3$\pm$ 0.7\%\\ \hline
(10)&96.4$\pm$ 0.5\%&97.0$\pm$ 0.2\%&96.9$\pm$ 0.4\%&{\bf 97.7$\pm$ 0.0\%}&97.4$\pm$ 0.2\%\\ \hline
(11)&74.4$\pm$ 0.7\%&75.1$\pm$ 0.6\%&73.3$\pm$ 2.3\%& {\bf 75.6$\pm$ 0.7\%} &73.5$\pm$ 2.6\%\\ \hline
(12)&89.2$\pm$ 0.6\%&89.3$\pm$ 0.3\%&90.5$\pm$ 0.6\%&91.6$\pm$ 0.1\%& {\bf 91.9$\pm$ 0.4\%}\\ \hline
(13)&95.6$\pm$ 0.7\%&96.0$\pm$ 0.3\%&94.8$\pm$ 0.9\%& {\bf 96.0$\pm$ 0.1\%} &94.8$\pm$ 0.3\%\\ \hline
(14)&94.3$\pm$ 0.5\%&94.2$\pm$ 0.3\%&94.6$\pm$ 0.5\%& {\bf 96.1$\pm$ 0.1\%} &96.1$\pm$ 0.5\%\\ \hline
(15)&97.4$\pm$ 0.5\%&97.0$\pm$ 0.3\%&97.4$\pm$ 0.3\%& {\bf 99.1$\pm$ 0.0\%} &98.3$\pm$ 0.1\%\\ \hline
(16)&91.8$\pm$ 0.7\%&90.4$\pm$ 0.5\%&92.2$\pm$ 0.8\%&93.9$\pm$ 0.2\%& {\bf 94.7$\pm$ 0.5\%} \\ \hline
(17)&96.9$\pm$ 0.2\%& {\bf 97.3$\pm$ 0.1\%} &96.3$\pm$ 0.3\%&94.1$\pm$ 0.0\%&93.7$\pm$ 0.2\%\\ \hline
(18)&92.1$\pm$ 0.4\%&92.0$\pm$ 0.2\%&92.5$\pm$ 0.5\%& {\bf 93.5$\pm$ 0.1\%} & {\bf 93.6$\pm$ 0.3\%}\\ \hline
(19)&84.7$\pm$ 0.6\%&84.5$\pm$ 0.5\%&86.1$\pm$ 0.6\%& {\bf 87.2$\pm$ 0.1\%} & {\bf 87.4$\pm$ 0.5\%}\\ \hline
(20)&97.0$\pm$ 0.5\%& {\bf 99.3$\pm$ 0.1\%} &96.9$\pm$ 0.6\%&97.1$\pm$ 0.1\%&95.1$\pm$ 0.2\%\\ \hline
\end{tabular}
\caption{Average revenue of first 50th rounds of five algorithms for semi-synthetic data sets with perfect response. The revenue is normalized by the optimal revenue of each data set. We repeat 50 trials to get the 95\% confidence interval.
\label{tab:semi-synthetic-revenue-50}}
\vspace*{-5pt}
\end{table}

\end{document}